\documentclass[12pt]{amsproc}

\usepackage{amsmath,amsfonts,amssymb}

\copyrightinfo{2000}{American Mathematical Society}

\newcommand{\cA}{{\mathcal A}}

\newcommand{\cD}{{\mathcal D}}
\newcommand{\cE}{{\mathcal E}}
\newcommand{\cH}{{\mathcal H}}
\newcommand{\cL}{{\mathcal L}}

\newcommand{\bN}{{\mathbb N}}
\newcommand{\bZ}{{\mathbb Z}}
\newcommand{\bQ}{{\mathbb Q}}
\newcommand{\bR}{{\mathbb R}}
\newcommand{\bC}{{\mathbb C}}

\newcommand{\ind}{\mathop{\rm ind}\nolimits}

\newtheorem{theorem}{Theorem}[section]
\newtheorem{lemma}{Lemma}[section]
\newtheorem{corollary}{Corollary}[section]
\newtheorem{definition}{Definition}[section]
\newtheorem{remark}{Remark}[section]

\newtheorem{proposition}{Proposition}[section]

\numberwithin{equation}{section}

\begin{document}

\title[Asymptotics of $p$-adic singular Fourier integrals]
{Asymptotical behavior of one class of $p$-adic singular Fourier integrals}

\author{A.~Yu.~Khrennikov}
\address{International Center for Mathematical Modelling in Physics
and Cognitive Sciences MSI, V\"axj\"o University,
SE-351 95, V\"axj\"o, Sweden.}
\email{andrei.khrennikov@msi.vxu.se}

\author{V.~M.~Shelkovich}
\address{Department of Mathematics, St.-Petersburg State Architecture
and Civil Engineering University, 2 Krasnoarmeiskaya 4, 190005,
St. Petersburg, Russia.}
\email{shelkv@vs1567.spb.edu}

\thanks{The second author (V.~S.) was supported in part by
DFG Project 436 RUS 113/809, DFG Project 436 RUS 113/951, and
Grant 05-01-04002-NNIOa of Russian Foundation for Basic Research.}


\date{}


\begin{abstract}
We study the asymptotical behavior of the $p$-adic singular Fourier integrals
$$
J_{\pi_{\alpha},m;\varphi}(t)
=\bigl\langle f_{\pi_{\alpha};m}(x)\chi_p(xt), \varphi(x)\bigr\rangle
=F\big[f_{\pi_{\alpha};m}\varphi\big](t), \quad |t|_p \to \infty,
\quad t\in \bQ_p,
$$
where $f_{\pi_{\alpha};m}\in {\cD}'(\bQ_p)$ is a {\em quasi associated
homogeneous} distribution (generalized function) of degree
$\pi_{\alpha}(x)=|x|_p^{\alpha-1}\pi_1(x)$ and order $m$, \
$\pi_{\alpha}(x)$, $\pi_1(x)$, and $\chi_p(x)$ are a
multiplicative, a normed multiplicative, and an additive
characters of the field $\bQ_p$ of $p$-adic numbers, respectively,
$\varphi \in {\cD}(\bQ_p)$ is a test function, $m=0,1,2\dots$,
$\alpha\in \bC$. If $Re\alpha>0$ the constructed asymptotics
constitute a $p$-adic version of the well known Erd\'elyi lemma.
Theorems which give asymptotic expansions of singular Fourier integrals
are the Abelian type theorems. In contrast to the real case,
all constructed asymptotics have the {\it stabilization\/} property.
\end{abstract}

\maketitle

\setcounter{equation}{0}

\section{Introduction}
\label{s1}

\subsection{$p$-Adic mathematical physics.}\label{s1.1}
According to the well-known Ostrovsky theorem, {\it any nontrivial
valuation on the field of the rational numbers $\bQ$ is equivalent
either to the real valuation $|\cdot|$ or to one of the $p$-adic
valuations $|\cdot|_p$}, where $p$ is any prime number.
This $p$-adic norm $|\cdot|_p$ is defined as follows:
if an arbitrary rational number $x\ne 0$ is represented as
$x=p^{\gamma}\frac{m}{n}$, where $\gamma=\gamma(x)\in \bZ$ and
the integers $m$, $n$ are not divisible by $p$, then
$$
|x|_p=p^{-\gamma}, \quad x\ne 0, \qquad |0|_p=0.
$$
The norm $|\cdot|_p$ satisfies the {\em strong triangle inequality}
\begin{equation}
\label{1.0}
|x+y|_p\le \max(|x|_p,|y|_p)
\end{equation}
and is non-Archimedean.
Consequently, it is possible to construct a completion of $\bQ$
only with respect to the real valuation $|\cdot|$ or to one of
the $p$-adic valuations $|\cdot|_p$.
The field $\bQ_p$ of $p$-adic numbers is defined as the
completion of the field of rational numbers $\bQ$ with respect
to the norm $|\cdot|_p$.

Thus there are two equal in rights universes: the ``real universe''
and the ``$p$-adic one''. The latter universe is non-Archimedean, and
in consequence of this has some specific and surprising properties.
This leads to interesting deviations from the classical ``real universe''.

As is well known, during a few hundred years theoretical physics has been
developed on the basis of real (and later also complex) numbers. However,
in the last 20 years the field of $p$-adic numbers ${\bQ}_p$ (as well as its
algebraic extensions) has been intensively used in theoretical and mathematical
physics, stochastics, psychology, cognitive and social sciences, biology,
image analysis (see~\cite{Av-Bik-Koz-O},~\cite{Kh1}--~\cite{Kh3},~\cite{Koch},
~\cite{Vl-V-Z}--~\cite{V2} and the references therein).
Thus, notwithstanding the fact that the $p$-adic numbers were discovered
by K.~Hensel around the end of the nineteenth century, the theory of
$p$-adic numbers has already penetrated into several areas of mathematics
and applied researches.

Since $p$-adic analysis and $p$-adic mathematical physics are young areas
there are many unsolved problems,  which have been solved in standard real
setting. Since ``$p$-adic universe'' is in a sense dual to the
``real universe'', solving such type of problems is important.

Recall that in the usual ($\bR$) case there is a theory of
so-called {\it oscillating integrals}, which have the form
$\int_{\bR^n}e^{itf(x)}\varphi(x)\,d^nx$. These integrals frequently
occur in applied and mathematical physics. The classical problem
related to {\it oscillating integrals\/} is to investigate their
asymptotical behavior when the parameter $t$ tends to
infinity~\cite{Arn-G-Var2},~\cite{Fed},~\cite[7.8]{Hor}. In
particular, there are many problems where solutions are obtained
as Fourier integrals which can not be evaluated exactly.
Nevertheless, these solutions are not less important because it is
often possible to study the asymptotic behavior of these
integrals~\cite[9]{J}.

The problem of the asymptotical behavior of the {\it Fourier integrals\/}
is related to the well known Erd\'elyi lemma~\cite{Erdelyi-1},~\cite{Erdelyi-2}.
In the one-dimensional case this lemma describes the asymptotics of the
Fourier transforms of functions $f(x)$ defined on $\bR$ and having
singularities of the type
$x_{\pm}^{\alpha-1}\log^mx_{\pm}\varphi(x)$, where $\alpha>0$ and
$\varphi(x)$ is sufficiently smooth~\cite[Ch.III,\S1]{Fed}. There
are multidimensional generalizations of this lemma~\cite[Ch.III]{Fed},
~\cite[Ch.II,\S7]{Arn-G-Var2},~\cite{Zasl}.

The above-mentioned problems are close to the problem of
constructing asymptotics of the Fourier transform of distributions
$$
(x\pm i0)^{\alpha-1}\log^m(x\pm i0)\varphi(x), \quad m=0,1,2\dots,
\quad \alpha\in \bC,
$$
where $\varphi(x)\in {\cD}(\bR)$ and $(x\pm i0)^{\alpha-1}\log^m(x\pm i0)$
are {\it quasi associated homogeneous distributions\/} of degree $\alpha-1$ and
order $m$ (see~\cite[Ch.I,\S 4.]{Ge-Sh},  and Remark~\ref{rem2}). These
asymptotics were constructed in~\cite{Brich-Shirok1},
~\cite{Brich-Shirok2},~\cite{Brich1} (see also~\cite[Ch.III,\S1.6.,\S8]{Fed}).
In these papers the following asymptotical formulas were derived:
$$
(x\pm i0)^{\alpha-1}\log^m(x\pm i0)e^{itx}
\qquad\qquad\qquad\qquad\qquad\qquad\qquad\qquad
$$
\begin{equation}
\label{2.1}
\approx
\delta(x)2\pi\sum_{k=0}^m\frac{d^{m-k}}{d\alpha^{m-k}}
\Big(\frac{e^{\pm i\frac{\pi(\alpha-1)}{2}}}{\Gamma(-\alpha+1)}\Big)
\frac{\log^{k}|t|}{|t|^{\alpha}}, \quad t\to \mp\infty,
\end{equation}
\begin{equation}
\label{2.2}
(x\pm i0)^{\alpha-1}\log^m(x\pm i0)e^{itx}=o(|t|^{-N}),
\quad t\to \pm\infty, \qquad\qquad
\end{equation}
for any $N\in \bN$, where $\alpha \not\in \bN$. Some particular
cases of these formulas were studied in~\cite[9]{J}.
In~\cite{Seller}, the asymptotic behavior of singular Fourier integrals
of pseudo-functions having power and logarithmic singularities are studied.

In $p$-adic analysis the above-mentioned problems have not been studied so far.
However, taking into account that $p$-adic mathematical physics is intensively
developed, studying these type of problems in $p$-adic setting is very
important.

\subsection{Contents of the paper.}\label{s1.2}
In this paper the asymptotical behavior of the {\em $p$-adic singular Fourier integrals}
\begin{equation}
\label{3}
J_{\pi_{\alpha},m;\varphi}(t)
=\bigl\langle f_{\pi_{\alpha};m}(x)\chi_p(xt),\varphi(x)\bigr\rangle
=F\big[f_{\pi_{\alpha};m}\varphi\big](t), \quad |t|_p \to \infty,
\end{equation}
is studied, where $f_{\pi_{\alpha};m}\in {\cD}'(\bQ_p)$ is a
{\em quasi associated homogeneous distribution} of degree
$\pi_{\alpha}(x)=|x|_p^{\alpha-1}\pi_1(x)$ and order $m$, \
$m=0,1,2\dots$, $\alpha\in \bC$ (see Definitions~\ref{de1},~\ref{de1.1}
and Theorem~\ref{th1.2}); $\varphi \in {\cD}(\bQ_p)$; $F$ is the Fourier
transform; $\pi_{\alpha}(x)$, $\pi_1(x)$, and $\chi_p(x)$ are a
multiplicative (\ref{16}), a normed multiplicative (\ref{16.1}),
and an additive characters of the field $\bQ_p$, respectively.

\begin{remark}
\label{rem1} \rm
(i) Let us note that the linear span of set of distributions mentioned above
$$
{\cA\cH}_0(\bR)={\rm span}\big\{(x\pm i0)^{\alpha}\log^m(x\pm i0):\, \alpha\in\bC, \, m\in \bN_0\big\}
\qquad\qquad
$$
$$
={\rm span}\big\{x_{\pm}^{\alpha}\log^k x_{\pm}, \ P\big(x_{\pm}^{-n}\log^{m-1}x_{\pm}\big):
\, \alpha\in\bC,
\quad
$$
$$
\qquad\qquad
\alpha \ne -1,-2,\dots,-n,\dots;
\,\, n, m \in \bN, \, k\in \bN_0 \big\}\subset \cD'(\bR)
$$
constitutes a class important for application in mathematical physics,
$\bN_0=\{0\}\cup\bN$.
Recall that this class was first introduced and studied in the
book~\cite[Ch.I,\S 4.]{Ge-Sh} as a class of the so-called
{\em associated homogeneous distributions}.
Later {\em associated homogeneous distributions} were studied in the
book~\cite{E-K2}.
Unfortunately, results on {\em associated homogeneous distribution} from
the books~\cite{Ge-Sh},~\cite{E-K2} are not quite
consistent and have self-contradictory (for details, see~\cite{Sh1}). The problems
of introducing of the concept of {\em associated homogeneous distribution}
for $\cD'(\bR^n)$ and relating mathematics were studied in~\cite{Sh1}.
According to~\cite{Sh1}, {\it direct transfer} of the notion of an
{\it associated eigenvector} to the case of distributions
{\it is impossible for the order $m\ge 2$}. Thus there exist only
{\em associated homogeneous distributions} of order $m=0$, i.e.,
{\em homogeneous distributions} (see Definition~\cite[Ch.I,\S 3.11.,(1)]{Ge-Sh},
~\cite[3.2.]{Hor}) and of order $m=1$ (see Definition~\cite[Ch.I,\S4.1.,(1),(2)]{Ge-Sh}).
Moreover, in~\cite{Sh1},a definition of {\it quasi associated homogeneous
distribution} which is a natural generalization of the notion of
{\it associated homogeneous distribution} was introduced and a mathematical
description of all quasi associated homogeneous distributions was given.
It was proved in~\cite{Sh1} that the class of {\it quasi associated homogeneous
distributions} coincides with the class of distributions ${\cA\cH}_0(\bR)$
introduced in~\cite[Ch.I,\S 4.]{Ge-Sh} as the class of {\it associated
homogeneous distributions}.

By adaptation of definitions from~\cite{Sh1} to the case of the field $\bQ_p$
(instead of the real field), a notion of the {\it $p$-adic quasi associated
homogeneous distribution} was introduced in~\cite{Al-Kh-Sh1},~\cite{Al-Kh-Sh2}
by Definition~\ref{de1.1}. (In~\cite{Al-Kh-Sh1},~\cite{Al-Kh-Sh2} these new
distributions were named as {\em associated homogeneous distributions}).
In~\cite{Al-Kh-Sh1},~\cite{Al-Kh-Sh2} a mathematical description of all
{\it $p$-adic quasi associated homogeneous distributions $f_{\pi_{\alpha};m}$}
was given (see Theorem~\ref{th1.2} and formulas (\ref{19.3}), (\ref{19.5})).

(ii) Note that associated homogeneous
distributions from ${\cD}'(\bR)$ are parametrized by $\alpha\in
\bC$ and $m\in \bN_0$, while associated homogeneous distributions
from ${\cD}'(\bQ_p)$ are parametrized by $\alpha\in \bC$,
$\pi_1(x)$, and $m\in \bN_0$ (cf.~\cite[Definition~3.3.]{Sh1} and
Definition~\ref{de1.1}).
\end{remark}

In Sec.~\ref{s1*}, some facts from $p$-adic theory of distributions
are presented. In particular, the results on {\em quasi associated
homogeneous distributions}~\cite{Al-Kh-Sh1},~\cite{Al-Kh-Sh2} are given
in Subsec.~\ref{s1*.3}.
In Sec.~\ref{s1**}, a definition of the {\it stable\/} $p$-adic asymptotical
expansion is introduced. This concept is relevant for $p$-adic asymptotic
analysis (see below).
In Sec.~\ref{s2},~\ref{s3}, we prove Theorems~\ref{th2-1},~\ref{th2-2},~\ref{th3}
which describe the asymptotical behavior of the $p$-adic singular Fourier
integrals (\ref{3}). Here the sequence
$$
\{|t|_p^{-\alpha}\pi_{1}^{-1}(t)\log_p^{m-k}|t|_p:k=0,1,2,\dots,m\}, \quad |t|_p \to \infty
$$
is an {\em asymptotic sequence}, and any coefficient of the asymptotic
expansion is proportional to the Dirac delta function. 
Here we note that, although statement
of Theorem~\ref{th2-1} can be obtained as a corollary of
Theorem~\ref{th3}, we prove these theorems separately to
demonstrate different methods for calculating $p$-adic
asymptotics. In Sec.~\ref{s4}, by Corollary~\ref{cor2} a
$p$-adic version of the well-known Erd\'elyi lemma is given. This
lemma is a direct consequence of Theorems~\ref{th2-1},~\ref{th3}
for the case $Re\alpha>0$. In Sec.~\ref{s5}, auxiliary lemmas are proved.

The asymptotical formulas (\ref{49.0})--(\ref{49.1}), (\ref{50*})--(\ref{50.1}),
(\ref{65})--(\ref{65.3}) obtained by Theorems~\ref{th2-1},~\ref{th2-2},~\ref{th3}
are $p$-adic analogs of formulas (\ref{2.1}), (\ref{2.2}).
However, in contrast to (\ref{2.1}), (\ref{2.2}), $p$-adic asymptotical formulas
(\ref{49.0})--(\ref{49.1}), (\ref{50*})--(\ref{50.1}), (\ref{65})--(\ref{65.3})
have a {\it specific property of the stabilization\/}. Namely, the left- and right-hand
sides of these formulas are {\it exact equalities\/} for sufficiently big
$|t|_p>s(\varphi)$, where $s(\varphi)$ is the stabilization parameter (see
Definition~\ref{de4} and Remark~\ref{rem2} ).
The stabilization parameter $s(\varphi)$ depends on the
{\it parameter of constancy} of the function $\varphi$ (see (\ref{9.3*}))
and the rank of the character $\pi_1(x)$ (see (\ref{17})).
Asymptotics of this type we call {\it stable\/} asymptotical expansions
(see Definitions~\ref{de2}--~\ref{de4}). This $p$-adic phenomenon is quite
different from the ``real asymptotic properties''. It was first discovered
in our paper~\cite[Theorem~5.1]{Al-Kh-Sh5}, where some weak asymptotics
were calculated.

This {\em asymptotic stabilization property\/} is similar
to another $p$-adic phenomenon: {\em if $\lim_{n\to\infty}x_n=x$, $x_n,x\in \bQ_p$,
$|x|_p\ne 0$, then $\lim_{n\to\infty}|x_n|_p=|x|_p$ and the sequence of norm
$\{|x_n|_p:n\in\bN\}$ must be {\em stabilize} for sufficiently large $n$}.
Indeed, since $|x_n-x|_p<|x|_p$ for sufficiently large $n$, according to
the {\em strong triangle inequality} (\ref{1.0}), we have
$$
|x_n|_p=|(x_n-x)+x|_p=\max(|x_n-x|_p,|x|_p)=|x|_p
\quad \text{for sufficiently large $n$}.
$$

It may well be that {\it stabilization\/} is a typical property of
$p$-adic asymptotics.

It remains to note that Theorems~\ref{th2-1},~\ref{th2-2},~\ref{th3}
are the {\em Abelian type theorems}. Theorems of this type are inverse
to the Tauberian theorems (see~\cite{Vl-D-Zav} and the references therein).
For the $p$-adic case Tauberian theorems for distributions were first
proved in~\cite{Al-Kh-Sh6},~\cite{Kh-Sh2}. In this paper we study the
asymptotical behavior of the singular Fourier integrals
$J_{\pi_{\alpha},m;\varphi}(t)=F[g(x)](t)$, where the functions
$g(x)=|x|_p^{\alpha-1}\log_p^{m}|x|_p\pi_{1}(x)\varphi(x)$ admit
the estimate $g(x)=O(|x|_p^{\alpha-1}\log_p^{m}|x|_p)$, $|x|_p \to 0$.
If $\alpha \ne 0$, according to Theorems~\ref{th2-1},~\ref{th3}, we have
$$
J_{\pi_{\alpha},m;\varphi}(t)=O\Big(|t|_p^{-\alpha}\log_p^{m}|t|_p\Big),
\quad |t|_p \to \infty.
$$
This connection between asymptotical behavior of $g(x)$ and
$J_{\pi_{\alpha},m;\varphi}(t)$ is a typical Abelian type theorem.

The results of this paper allow a development of an area of $p$-adic
harmonic analysis which has not been studied so far. In addition, a
new technique of constructing $p$-adic weak asymptotics is developed.
Moreover, a new effect of the $p$-adic {\em asymptotic stabilization} is observed.

Since the asymptotical formulas for the Fourier transform of
quasi associated homogeneous distributions from ${\cD}'(\bR)$ have many
applications (see, for example~\cite{Brich-Shirok2},
 \cite{Shirok1},~\cite{Zasl}), we hope that their $p$-adic
versions may be also useful in the $p$-adic mathematical physics.

\section{Preliminary results in $p$-adic analysis}
\label{s1*}

\subsection{$p$-Adic functions and distributions.}\label{s1*.2}
We shall use intensively the notations and results from~\cite{Vl-V-Z}.
Denote by $\bN$, $\bZ$, $\bC$ the sets of positive integers,
integers, complex numbers, respectively, and set
${\bN}_0={0}\cup{\bN}$. Denote by $\bQ_p^{*}=\bQ_p\setminus\{0\}$
the multiplicative group of the field $\bQ_p$.

Denote by $B_{\gamma}(a)=\{x\in\bQ_p:|x-a|_p\le p^{\gamma}\}$ the ball
of radius $p^{\gamma}$ with center at a point $a\in \bQ_p$ and by
$S_{\gamma}(a)=\{x\in\bQ_p:|x-a|_p=p^{\gamma}\}
=B_{\gamma}(a)\setminus B_{\gamma-1}(a)$ its boundary (sphere),
$\gamma \in \bZ$. For $a=0$ we set $B_{\gamma}(0)=B_{\gamma}$ and
$S_{\gamma}(0)=S_{\gamma}$.

On $\bQ_p$ one can define the Haar measure, i.e., a positive
measure $dx$ which is invariant with respect to shifts,
$d(x+a)=dx$, and normalized by the equality $\int_{|\xi|_p\le 1}\,dx=1$.

A complex-valued function $f$ defined on $\bQ_p$ is called {\it locally-constant}
if for any $x\in \bQ_p$ there exists an integer $l(x)\in \bZ$ such that
$$
f(x+x')=f(x), \quad x'\in B_{l(x)}.
$$

Let ${\cE}(\bQ_p)$ and ${\cD}(\bQ_p)$ be the linear spaces of
locally-constant $\bC$-valued functions on $\bQ_p$ and locally-constant
$\bC$-valued functions with compact supports (so-called test functions),
respectively. According to Lemma~1 from~\cite[VI.1.]{Vl-V-Z}, for any
$\varphi \in {\cD}(\bQ_p)$ there exists $l\in \bZ$, such that
\begin{equation}
\label{9.3*}
\varphi(x+x')=\varphi(x), \quad x'\in B_l, \quad x\in \bQ_p.
\end{equation}
The largest number $l=l(\varphi)$ for which the last relation
holds is called the {\it parameter of constancy} of the function $\varphi$.
Let us denote by ${\cD}^l_N(\bQ_p)$ the space of test
functions from ${\cD}(\bQ_p)$ with supports in the disc $B_N$ and with
parameter of constancy $\ge l$. The following embedding holds:
${\cD}^l_N(\bQ_p) \subset {\cD}^{l'}_{N'}(\bQ_p)$, \ $N\le N'$, $l\ge l'$.
Here ${\cD}(\bQ_p)=\lim\limits_{N\to \infty}\ind{\cD}_N(\bQ_p)$,
${\cD}_N(\bQ_p)=\lim\limits_{l\to -\infty}\ind{\cD}_N^l(\bQ_p)$.
Denote by ${\cD}'(\bQ_p)$ the set of all
linear functionals on ${\cD}(\bQ_p)$.

Denote by $\Delta_k(x)\stackrel{def}{=}\Omega(p^{-k}|x|_p)$
the characteristic function of the ball $B_{k}$, $k \in \bZ$, \
$x\in \bQ_p$, where
$$
\Omega(t)=\left\{
\begin{array}{rclrcl}
\displaystyle
&1,& \quad 0 \le &t& \le 1, \\
\displaystyle
&0,& \quad  &t& > 1. \\
\end{array}
\right.
$$

If $f\in{\cD}'(\bQ_p)$, $\varphi\in {\cD}(\bQ_p)$ then the
convolution $f*\varphi\in {\cE}(\bQ_p)$ and~\cite[VII,(1.7)]{Vl-V-Z}
\begin{equation}
\label{11.1}
(f*\varphi)(x)=\langle f(\xi),\varphi(x-\xi)\rangle.
\end{equation}

The Fourier transform of a test function $\varphi\in {\cD}(\bQ_p)$
is defined by the formula
$$
F[\varphi](\xi)=\int \chi_p(\xi x)\varphi(x)\,dx, \quad \xi \in \bQ_p,
$$
where the function $\chi_p(\xi x)=e^{2\pi i\{\xi x\}_p}$ for every
fixed $\xi \in \bQ_p$ is an {\it additive character\/} of the
field $\bQ_p$, \ $\{\xi x\}_p$ is the fractional part of the
number $\xi x$~\cite[VII.2.,3.]{Vl-V-Z}.
\begin{lemma}
\label{lem1}
{\rm(}{\rm~\cite[III,(3.2)]{Taib3},~\cite[VII.2.]{Vl-V-Z}}{\rm)}
Fourier transform is a linear isomorphism ${\cD}(\bQ_p)$ into ${\cD}(\bQ_p)$.
Moreover,
\begin{equation}
\label{12.1}
\varphi \in {\cD}^l_N(\bQ_p) \quad \text{iff}
\quad
F[\varphi] \in {\cD}^{-N}_{-l}(\bQ_p).
\end{equation}
\end{lemma}

We define the Fourier transform $F[f]$ of a distribution
$f\in{\cD}'(\bQ_p)$ by the relation~\cite[VII.3.]{Vl-V-Z}:
$\langle F[f],\varphi\rangle=\langle f,F[\varphi]\rangle$,
for all $\varphi\in {\cD}(\bQ_p)$.

Any {\it multiplicative character\/} (see~\cite[III.2.]{Vl-V-Z})
$\pi$ of the field $\bQ_p$ can be represented as
\begin{equation}
\label{16}
\pi(x)\stackrel{def}{=}\pi_{\alpha}(x)=|x|_p^{\alpha-1}\pi_{1}(x),
\quad x \in \bQ_p,
\end{equation}
where $\pi(p)=p^{1-\alpha}$ and $\pi_{1}(x)$ is a
{\it normed multiplicative character\/} such that
\begin{equation}
\label{16.1}
\pi_1(x)=\pi_{1}(|x|_px), \quad \pi_1(p)=\pi_1(1)=1,
\quad |\pi_1(x)|=1.
\end{equation}
We denote $\pi_{0}=|x|_p^{-1}$.

\begin{lemma}
\label{lem2}
{\rm(}{\rm~\cite[I.7.]{Taib3},~\cite[III,(2.2)]{Vl-V-Z}}{\rm)}
Let $A_0=S_0=\{x\in \bQ_p:|x|_p=1\}$,
$A_{k}=B_{-k}(1)=\{x\in \bQ_p:|x-1|_p\le p^{-k}\}$, $k\in\bN$.
If $\pi_1$ is a normed multiplicative character {\rm(\ref{16.1})},
then there exists $k\in {\bN}_0$ such that
\begin{equation}
\label{17}
\pi_1(x)\equiv 1, \quad x\in A_{k}.
\end{equation}
\end{lemma}

The smallest $k_0\in {\bN}_0$ for which the equality
(\ref{17}) holds is called the {\it rank of the normed multiplicative
character} $\pi_1(x)$. There is only one {\em zero rank character}, namely,
$\pi_1(x)\equiv 1$.

Let us introduce the $p$-adic $\Gamma$-{\it functions\/}
(see~\cite[VIII,(2.2),(2.17)]{Vl-V-Z}):
\begin{equation}
\label{25}
\Gamma_p(\alpha)\stackrel{def}{=}\Gamma_p(|x|_p^{\alpha-1})
=\int_{\bQ_p} |x|_p^{\alpha-1}\chi_p(x)\,dx
=\frac{1-p^{\alpha-1}}{1-p^{-\alpha}};
\end{equation}
\begin{equation}
\label{25.1}
\Gamma_p(\pi_{\alpha})\stackrel{def}{=}F[\pi_{\alpha}](1)
=\int_{\bQ_p}|x|_p^{\alpha-1}\pi_{1}(x)\chi_p(x)\,dx.
\end{equation}
Here the integrals in the right-hand sides of (\ref{25}), (\ref{25.1})
are defined by means of analytic continuation with respect to $\alpha$.
According to~\cite[III,Theorem (4.2)]{Taib3},
$\Gamma$-function (\ref{25.1}) can be also defined as improper integral
$\lim_{k\to\infty}\int_{p^{-k}\le |x|_p\le p^{k}}\,\cdot\,\,dx$.

\subsection{Homogeneous and quasi associated homogeneous distributions.}\label{s1*.3}
Let us recall some facts on $p$-adic homogeneous and quasi associated homogeneous
distributions.

\begin{definition}
\label{de1} \rm
(~\cite[Ch.II,\S 2.3.]{G-Gr-P},~\cite[VIII.1.]{Vl-V-Z})
Let $\pi_{\alpha}$ be a multiplicative character (\ref{16}) of
the field $\bQ_p$. A distribution $f \in {\cD}'(\bQ_p)$ is called
{\it homogeneous} of degree $\pi_{\alpha}$ if for all
$\varphi \in {\cD}(\bQ_p)$ we have
$$
\Bigl\langle f,\varphi\Big(\frac{x}{t}\Big) \Bigr\rangle
=\pi_{\alpha}(t)|t|_p \langle f,\varphi \rangle, \quad \forall \, t \in \bQ_p^*,
$$
i.e., $f(tx)=\pi_{\alpha}(t)f(x)$, $t \in \bQ_p^{*}$.
\end{definition}

The following theorem gives a description of all {\em homogeneous distributions}.

\begin{theorem}
\label{th1} \rm{(~\cite[Ch.II,\S
2.3.]{G-Gr-P},~\cite[VIII.1.]{Vl-V-Z})} Every homogeneous distribution
$f\in {\cD}'(\bQ_p)$ of degree $\pi_{\alpha}$ has the form

{\rm (a)} $C\pi_{\alpha}$ if $\pi_{\alpha}\ne\pi_{0}=|x|_p^{-1}$;

{\rm (b)} $C\delta$ if $\pi_{\alpha}=\pi_{0}=|x|_p^{-1}$, where $C$ is a
constant.
\end{theorem}

\begin{definition}
\label{de1.1}
\rm (~\cite{Al-Kh-Sh1},~\cite{Al-Kh-Sh2}) A distribution
$f_m\in {\cD}'(\bQ_p)$ is said to be {\it quasi associated
homogeneous} of
degree~$\pi_{\alpha}$ and order~$m$, \ $m \in \bN_{0}$, if
for all $\varphi \in {\cD(\bQ_p)}$ we have
$$
\Bigl\langle f_m,\varphi\Big(\frac{x}{t}\Big)\Bigr\rangle
=\pi_{\alpha}(t)|t|_p \langle f_m,\varphi \rangle
+\sum_{j=1}^{m}\pi_{\alpha}(t)|t|_p\log_p^j|t|_p
\langle f_{m-j},\varphi \rangle, \quad \forall \, t \in \bQ_p^*,
$$
where $f_{m-j}\in {\cD}'(\bQ_p)$ is an associated homogeneous distribution
of degree~$\pi_{\alpha}$ and order $m-j$, \ $j=1,2,\dots,m$, i.e.,
$$
f_m(tx)=\pi_{\alpha}(t)f_m(x)
+\sum_{j=1}^{m}\pi_{\alpha}(t)\log_p^j|t|_pf_{m-j}(x).
$$
If $m=0$ we set that the above sum is empty.
\end{definition}

A class of {\it quasi associated homogeneous} distributions of order~$m=0$
coincides with the class of {\it homogeneous} distributions.

\begin{theorem}
\label{th1.2}
\rm{(~\cite{Al-Kh-Sh1},~\cite{Al-Kh-Sh2})} Every
associated homogeneous distribution $f\in {\cD}'(\bQ_p)$ of degree
$\pi_{\alpha}(x)$ and order $m\in \bN$
{\rm(}with accuracy up to an associated homogeneous distribution
of order $\le m-1${\rm)} has the form

{\rm (a)} $C\pi_{\alpha}(x)\log_p^m|x|_p$ if  $\pi_{\alpha}\ne\pi_{0}=|x|_p^{-1}$;

{\rm (b)} $CP\big(|x|_p^{-1}\log_p^{m-1}|x|_p\big)$ if
 $\pi_{\alpha}=\pi_{0}=|x|_p^{-1}$, where $C$ is a constant.
\end{theorem}

According to the papers~\cite{Al-Kh-Sh1},~\cite{Al-Kh-Sh2}, an
{\it associated homogeneous distribution\/} of degree
$\pi_{\alpha}(x)=|x|_p^{\alpha-1}\pi_1(x) \ne \pi_{0}(x)=|x|_p^{-1}$
and order $m\in \bN$ is defined as
$$
\langle \pi_{\alpha}(x)\log_p^m|x|_p,\varphi(x) \rangle
=\int_{B_0}|x|_p^{\alpha-1}\pi_1(x)\log_p^m|x|_p
\big(\varphi(x)-\varphi(0)\big)\,dx
$$
$$
+\int_{\bQ_p\setminus B_0}|x|_p^{\alpha-1}\pi_1(x)\log_p^m|x|_p\varphi(x)\,dx
\qquad\qquad
$$
\begin{equation}
\label{19.3}
\qquad\qquad\qquad\qquad
+\varphi(0)\int_{B_0}|x|_p^{\alpha-1}\pi_1(x)\log_p^m|x|_p\,dx,
\end{equation}
for all $\varphi\in {\cD}(\bQ_p)$, where
$$
I_{0}(\alpha;m)=\int_{B_0}|x|_p^{\alpha-1}\pi_1(x)\log_p^m|x|_p\,dx
\qquad\qquad\qquad\qquad\qquad\qquad\qquad\qquad
$$
\begin{equation}
\label{19.4}
=\log_p^m e\frac{d^m I_{0}(\alpha)}{d\alpha^m}
=\log_p^m e\left\{
\begin{array}{rcl}
0, \quad \pi_1(x) &\not\equiv& 1, \\
\frac{d^m}{d\alpha^m}
\bigg(\frac{1-p^{-1}}{1-p^{-\alpha}}\bigg), \quad \pi_1(x) &\equiv& 1, \\
\end{array}
\right.
\end{equation}
where the integral
\begin{equation}
\label{24.1}
I_0(\alpha)=\int_{B_0}|x|_p^{\alpha-1}\pi_1(x)\,dx
=\left\{
\begin{array}{rcl}
0, \quad \pi_1(x) &\not\equiv& 1, \\
\frac{1-p^{-1}}{1-p^{-\alpha}}, \quad \pi_1(x) &\equiv& 1, \\
\end{array}
\right.
\end{equation}
is well defined for $Re \alpha >0$, and for
$\alpha \ne \alpha_j=\frac{2\pi i}{\ln p}j$, $j\in \bZ$ (\ref{24.1})
is defined by means of analytic continuation.

According to the same papers, an {\it associated homogeneous
distribution\/} of degree $\pi_{0}(x)=|x|_p^{-1}$ and order $m\in
\bN$ is defined as
$$
\Bigl\langle P\Big(\frac{\log_p^{m-1}|x|_p}{|x|_p}\Big),\varphi
\Bigr\rangle
\qquad\qquad\qquad\qquad\qquad\qquad\qquad\qquad\qquad
$$
\begin{equation}
\label{19.5}
=\int_{B_0}\frac{\log_p^{m-1}|x|_p}{|x|_p}\big(\varphi(x)-\varphi(0)\big)\,dx
+\int_{\bQ_p\setminus
B_0}\frac{\log_p^{m-1}|x|_p}{|x|_p}\varphi(x)\,dx,
\end{equation}
for all $\varphi\in {\cD}(\bQ_p)$.

\section{$p$-Adic stable distributional asymptotics}
\label{s1**}

Let us introduce a definition of the distributional
asymptotics~\cite{Brich} adapted to the case of $\bQ_p$.

\begin{definition}
\label{de2} \rm A sequence of continuous complex-valued functions
$\psi_k(t)$ on the multiplicative group $\bQ_p^*$ is called an
asymptotic sequence, as $|t|_p \to \infty$ if
$\psi_{k+1}(t)=o(\psi_k(t))$, $|t|_p \to \infty$ for all
$k=1,2,\dots$.
\end{definition}

\begin{definition}
\label{de3}
\rm Let $f(x,t)\in {\cD}'(\bQ_p)$ be a distribution
depending on $t$ as a parameter, and $C_k(x)\in {\cD}'(\bQ_p)$ be
distributions, $k=1,2,\dots$. We say that the relation
\begin{equation}
\label{4}
f(x,t)\approx \sum_{k=1}^{\infty}C_k(x)\psi_k(t),
\quad |t|_p \to \infty,
\end{equation}
is an asymptotical expansion of the distribution $f(x,t)$, as
$|t|_p \to \infty$, with respect to an asymptotic sequence
$\{\psi_k(t)\}$ if
\begin{equation}
\label{4.1}
\langle f(x,t),\varphi(x) \rangle \approx
\sum_{k=1}^{\infty}\langle C_k(x), \varphi(x) \rangle\psi_k(t),
\quad |t|_p \to \infty,
\end{equation}
for any $\varphi\in {\cD}(\bQ_p)$, i.e.,
$$
\langle f(x,t),\varphi(x) \rangle -\sum_{k=1}^{N}\langle C_k(x),
\varphi(x) \rangle\psi_k(t)=o(\psi_N(t)), \quad |t|_p \to \infty,
$$
for any $N$.
\end{definition}

\begin{definition}
\label{de4} \rm
Suppose that a distribution $f(x,t)\in
{\cD}'(\bQ_p)$ has the asymptotical expansion (\ref{4}). If for
any test function $\varphi(x)\in {\cD}(\bQ_p)$ there exists a
number $s(\varphi)$ depending on~$\varphi$ such that for all
$|t|_p>s(\varphi)$ relation (\ref{4.1}) is an {\it exact
equality\/}, we say that the asymptotical expansion (\ref{4}) is
{\it stable\/} and write
\begin{equation}
\label{4.2}
f(x,t)=\sum_{k=1}^{\infty}C_k(x)\psi_k(t), \quad |t|_p \to \infty.
\end{equation}
A number $s(\varphi)$ is called the {\it stabilization parameter}
of the asymptotical expansion (\ref{4}).
\end{definition}

\section{Asymptotic formulas for singular Fourier integrals (the case $\pi_1(x) \equiv 1$)}
\label{s2}

\subsection{The case $f_{\pi_{\alpha};m}(x)=|x|_p^{\alpha-1}\log_p^m|x|_p$, $\alpha\ne 0$,
$m=0,1,2,\dots$.}\label{s2.1}

\begin{theorem}
\label{th2-1}
Let $\varphi\in {\cD}^l_N(\bQ_p)$. Then the functional
$J_{\pi_{\alpha},m;\varphi}(t)$ has the following asymptotical
behavior:

{\rm (a)} If $m=0$, then
\begin{equation}
\label{49.0}
J_{\pi_{\alpha},0;\varphi}(t) =\Bigl\langle
|x|_p^{\alpha-1}\chi_p(xt),\varphi(x)\Bigr\rangle
=\varphi(0)\frac{\Gamma_p(\alpha)}{|t|_p^{\alpha}}, \quad |t|_p > p^{-l},
\end{equation}
the $\Gamma$-function $\Gamma_p(\alpha)$ is given by {\rm (\ref{25})},
i.e., in the weak sense
\begin{equation}
\label{49.0*}
\quad
|x|_p^{\alpha-1}\chi_p(xt)
=\delta(x)\frac{\Gamma_p(\alpha)}{|t|_p^{\alpha}}, \quad |t|_p \to
\infty.
\end{equation}

{\rm (b)} If $m=1,2,\dots$, then
$$
J_{\pi_{\alpha},m;\varphi}(t) =\Bigl\langle
|x|_p^{\alpha-1}\log_p^m|x|_p\chi_p(xt), \varphi(x)\Bigr\rangle
\qquad\qquad\qquad\qquad
$$
\begin{equation}
\label{49}
\qquad
=\varphi(0)\sum_{k=0}^{m}C_{m}^{k}\log_p^{k}e\,\frac{d^{k}\Gamma_p(\alpha)}{d\alpha^{k}}
\frac{\log_p^{m-k}|t|_p}{|t|_p^{\alpha}}, \quad |t|_p > p^{-l},
\end{equation}
i.e., in the weak sense
$$
|x|_p^{\alpha-1}\log_p^m|x|_p\chi_p(xt)
\qquad\qquad\qquad\qquad\qquad\qquad\qquad
$$
\begin{equation}
\label{49.1}
\qquad
=\delta(x)\sum_{k=0}^{m}C_{m}^{k}\log_p^{k}e\,\frac{d^{k}\Gamma_p(\alpha)}{d\alpha^{k}}
\frac{\log_p^{m-k}|t|_p}{|t|_p^{\alpha}}, \quad |t|_p \to \infty,
\end{equation}
with respect to an asymptotic sequence
$\{|t|_p^{-\alpha}\log_p^{m-k}|t|_p:k=0,1,\dots,m\}$.

Thus for any $\varphi \in {\cD}(\bQ_p)$, relations {\rm(\ref{49.0})},
{\rm (\ref{49})} are exact equalities for sufficiently big
$|t|_p>p^{-l}$, i.e., these asymptotical expansions are stable
with the stabilization parameter $s(\varphi)=p^{-l}$.
\end{theorem}

\begin{proof}
Let $Re \alpha>0$. In this case
$|x|_p^{\alpha-1}\log_p^m|x|_p\varphi(x)\in {\cL}^{1}(\bQ_p)$, and
the integral
$$
J_{\pi_{\alpha},m;\varphi}(t) =\bigl\langle
|x|_p^{\alpha-1}\log_p^m|x|_p\chi_p(xt), \varphi(x)\bigr\rangle
\qquad\qquad
$$
$$
\qquad
=\int_{\bQ_p}
|x|_p^{\alpha-1}\log_p^m|x|_p\chi_p(xt)\varphi(x)\,dx
$$
converges absolutely. Hence, according to the Riemann-Lebesque
theorem \cite[VII.3.]{Vl-V-Z}, $J_{\pi_{\alpha},m;\varphi}(t)\to 0$,
as $|t|_p \to \infty$. More precisely, since $\varphi(x)\in
{\cD}^l_N(\bQ_p)$ then, in view of Lemmas~\ref{lem3},~\ref{lem4},
\begin{equation}
\label{51}
J_{\pi_{\alpha},m;\varphi}(t)
=\varphi(0)\frac{1}{|t|_p^{\alpha}}
\sum_{k=0}^{m}C_{m}^{k}\log_p^{k}e\,\frac{d^{k}\Gamma_p(\alpha)}{d\alpha^{k}}
\log_p^{m-k}|t|_p, \quad \forall \, |t|_p> p^{-l}.
\end{equation}
Thus relations (\ref{49.0})--(\ref{49.1}) hold.

Let $Re \alpha <0$. In this case we define the functional
$J_{\pi_{\alpha},m;\varphi}(t)$ by the analytical continuation
with respect to $\alpha$. According to (\ref{19.3}), (\ref{19.4}):
$$
J_{\pi_{\alpha},m;\varphi}(t) =\bigl\langle
|x|_p^{\alpha-1}\log_p^m|x|_p\chi_p(xt), \varphi(x)\bigr\rangle
\qquad\qquad\qquad\qquad\qquad
$$
$$
=\int_{B_0}|x|_p^{\alpha-1}\log_p^m|x|_p\chi_p(xt)
\big(\varphi(x)-\varphi(0)\big)\,dx
$$
$$
\qquad\qquad\quad
+\int_{\bQ_p\setminus B_0}
|x|_p^{\alpha-1}\log_p^m|x|_p\chi_p(xt)\varphi(x)\,dx
$$
\begin{equation}
\label{45}
\qquad\quad
+\varphi(0)\int_{B_0}|x|_p^{\alpha-1}\log_p^m|x|_p\chi_p(xt)\,dx,
\end{equation}
for all $\varphi\in {\cD}(\bQ_p)$, where the last integral in
(\ref{45}) is defined by means of analytic continuation with
respect to $\alpha$.

Since $\varphi\in {\cD}^l_N(\bQ_p)$, it is {\it natural\/} to
rewrite functional (\ref{45}) as the following sum:
\begin{equation}
\label{45.1}
J_{\pi_{\alpha},m;\varphi}(t)=J_{\pi_{\alpha},m;\varphi}^{1}(t)
+J_{\pi_{\alpha},m;\varphi}^{2}(t)+\varphi(0)J_{\pi_{\alpha},m}^{0}(t),
\end{equation}
where
\begin{equation}
\label{46}
J_{\pi_{\alpha},m;\varphi}^{1}(t)
=\int_{B_l}|x|_p^{\alpha-1}\log_p^m|x|_p\chi_p(xt)
\big(\varphi(x)-\varphi(0)\big)\,dx,
\end{equation}
\begin{equation}
\label{47}
J_{\pi_{\alpha},m;\varphi}^{2}(t) =\int_{\bQ_p\setminus
B_l} |x|_p^{\alpha-1}\log_p^m|x|_p\chi_p(xt)\varphi(x)\,dx,
\end{equation}
\begin{equation}
\label{48}
J_{\pi_{\alpha},m}^{0}(t)
=\int_{B_l}|x|_p^{\alpha-1}\log_p^m|x|_p\chi_p(xt)\,dx.
\qquad\qquad
\end{equation}
Here integral (\ref{48}) is defined by means of analytic
continuation with respect to $\alpha$.

For $Re \alpha >0$ and $m=0$, according to (\ref{29}), integral
(\ref{48}) is equal to
$$
J_{\pi_{\alpha},0}^{0}(t)=F\big[|x|_p^{\alpha-1}\Delta_{l}(x)\big](t)
=\int_{B_{l}}\chi_p(tx)|x|_p^{\alpha-1}\,dx \qquad\qquad
$$
\begin{equation}
\label{52}
\qquad\qquad
=\frac{1-p^{-1}}{1-p^{-\alpha}}p^{\alpha l}\Delta_{-l}(t)
+\frac{\Gamma_p(\alpha)}{|t|_p^{\alpha}}\big(1-\Delta_{-l}(t)\big).
\end{equation}
For any $\alpha \ne \alpha_j=\frac{2\pi i}{\ln p}j$, \ $j\in \bZ$
we define $J_{\pi_{\alpha},0}^{0}(t)$ by means of analytic
continuation with respect to $\alpha$.

Differentiating relation (\ref{52}) with respect to~$\alpha$, we
obtain
$$
J_{\pi_{\alpha},m}^{0}(t)
=F\big[|x|_p^{\alpha-1}\log_p^m|x|_p\Delta_{l}(x)\big](t)
\qquad\qquad\qquad\qquad\qquad
$$
$$
=\int_{B_{l}}\chi_p(tx)|x|_p^{\alpha-1}\log_p^m|x|_p\,dx
=\log_p^m e\frac{d^{m}}{d\alpha^{m}}J_{\pi_{\alpha},0}^{0}(t)
$$
$$
=\Delta_{-l}(t)\big(1-p^{-1}\big)
\frac{d^m}{d\alpha^m}\Big(\frac{p^{\alpha l}}{1-p^{-\alpha}}\Big)\log_p^m e
\qquad\qquad\quad
$$
\begin{equation}
\label{53}
\qquad\qquad
+\big(1-\Delta_{-l}(t)\big)\frac{1}{|t|_p^{\alpha}}
\sum_{k=0}^{m}C_{m}^{k}\log_p^{k}e\,\frac{d^{k}\Gamma_p(\alpha)}{d\alpha^{k}}
\log_p^{m-k}|t|_p.
\end{equation}
Note that by using formulas from~\cite[Ch.II,\S2.2.]{G-Gr-P},~\cite[IV]{Vl-V-Z},
relation (\ref{53}) can be calculated explicitly.

According to (\ref{53}), we have
\begin{equation}
\label{54}
J_{\pi_{\alpha},m}^{0}(t)=\frac{1}{|t|_p^{\alpha}}
\sum_{k=0}^{m}C_{m}^{k}\frac{d^{k}\Gamma_p(\alpha)}{d\alpha^{k}}
\log_p^{m-k}|t|_p, \quad |t|_p> p^{-l}.
\end{equation}

Since $\varphi\in {\cD}^l_N(\bQ_p)$, it is clear that the
functions
$$
\begin{array}{rcl}
\displaystyle
|x|_p^{\alpha-1}\log_p^m|x|_p\big(\varphi(x)-\varphi(0)\big)\Delta_{l}(x)&=&0,
\medskip \\
\displaystyle
|x|_p^{\alpha-1}\log_p^m|x|_p\varphi(x)\big(1-\Delta_{l}(x)\big)&\in& {\cD}^l_N(\bQ_p).
\end{array}
$$
Thus for their Fourier transforms, according to (\ref{12.1}), we
have
\begin{equation}
\label{55}
\begin{array}{rcl}
\displaystyle
J_{\pi_{\alpha},m;\varphi}^{1}(t)&=&
\displaystyle
\int_{B_l}|x|_p^{\alpha-1}\log_p^m|x|_p\chi_p(xt)
\big(\varphi(x)-\varphi(0)\big)\,dx=0, \medskip \\
\displaystyle
J_{\pi_{\alpha},m;\varphi}^{2}(t)&=&
\displaystyle
\int_{\bQ_p\setminus B_l}
|x|_p^{\alpha-1}\log_p^m|x|_p\chi_p(xt)\varphi(x)\,dx=0, \\
\end{array}
\end{equation}
for all $|t|_p> p^{-l}$. Thus for $Re \alpha <0$ relations
(\ref{45.1}), (\ref{55}), (\ref{54}) imply
(\ref{49.0})--(\ref{49.1}).
\end{proof}

\subsection{The case $f_{\pi_{0};m}(x)=P\Big(\frac{\log_p^m|x|_p}{|x|_p}\Big)$,
$m=0,1,2,\dots$.}\label{s2.2}

\begin{theorem}
\label{th2-2}
Let $\varphi\in {\cD}^l_N(\bQ_p)$. Then the functional
$J_{\pi_{\alpha},m;\varphi}(t)$ has the following asymptotical
behavior:

{\rm (a)} If $m=0$ then
$$
J_{\pi_{0},0;\varphi}(t) =\Bigl\langle
P\Big(\frac{1}{|x|_p}\Big)\chi_p(xt), \varphi(x)\Bigr\rangle
\qquad\qquad\qquad\qquad\qquad
$$
\begin{equation}
\label{50*}
\qquad\quad
=\varphi(0)\bigg(-\frac{1}{p}-\Big(1-\frac{1}{p}\Big)
\log_p\Big(\frac{|t|_p}{p^{-l}}\Big)\bigg), \quad |t|_p > p^{-l},
\end{equation}
i.e., in the weak sense
\begin{equation}
\label{50.1*}
P\Big(\frac{1}{|x|_p}\Big)\chi_p(xt)
=\delta(x)\bigg(-\frac{1}{p}-\Big(1-\frac{1}{p}\Big)
\log_p\Big(\frac{|t|_p}{p^{-l}}\Big)\bigg), \quad |t|_p \to
\infty.
\end{equation}

{\rm (b)} If $m=1,2,\dots$ then
$$
J_{\pi_{0},m;\varphi}(t) =\Bigl\langle
P\Big(\frac{\log_p^m|x|_p}{|x|_p}\Big)\chi_p(xt),
\varphi(x)\Bigr\rangle \qquad\qquad\qquad\qquad
$$
$$
=\varphi(0)\bigg\{\frac{1}{p}(-1)^{m+1}(\log_p|t|_p-1)^m
\qquad\qquad
$$
$$
+\Big(1-\frac{1}{p}\Big)
\frac{1}{m+1}\Big((-1)^{m+1}\big(\log_p^{m+1}|t|_p-(-l)^{m+1}\big)
$$
$$
\qquad\qquad -(-1)^mC_{m+1}^1{\bf B}_{1}\big(\log_p^{m}|t|_p-(-l)^{m}\big)
$$
\begin{equation}
\label{50}
+\sum_{r=2}^{m}(-1)^{m+1-r}C_{m+1}^r{\bf B}_{r}
\big(\log_p^{m+1-r}|t|_p-(-l)^{m+1-r}\big)\Big)\bigg\}, \quad
|t|_p > p^{-l},
\end{equation}
where the Bernoulli numbers ${\bf B}_{r}$, $r=0,1,\dots,m$ are
defined by {\rm (\ref{131})}, i.e., in the weak sense,
$$
P\Big(\frac{\log_p^m|x|_p}{|x|_p}\Big)\chi_p(xt)
=\delta(x)\bigg\{-\frac{1}{p}\big(-\log_p|t|_p+1\big)^m
\qquad\qquad\qquad
$$
$$
+\Big(1-\frac{1}{p}\Big)
\frac{1}{m+1}\Big((-1)^{m+1}\big(\log_p^{m+1}|t|_p-(-l)^{m+1}\big)
$$
$$
\qquad\qquad
-(-1)^mC_{m+1}^1{\bf B}_{1}\big(\log_p^{m}|t|_p-(-l)^{m}\big)
$$
\begin{equation}
\label{50.1}
+\sum_{r=2}^{m}(-1)^{m+1-r}C_{m+1}^r{\bf B}_{r}
\big(\log_p^{m+1-r}|t|_p-(-l)^{m+1-r}\big)\Big)\bigg\},
\quad |t|_p \to \infty,
\end{equation}
with respect to an asymptotic sequence
$\{\log_p^{m+1-k}|t|_p:k=0,1,\dots,m+1\}$.

Thus for any $\varphi \in {\cD}(\bQ_p)$, relations {\rm (\ref{50*})},
{\rm (\ref{50})} are exact equalities for sufficiently big
$|t|_p>p^{-l}$, i.e., these asymptotical expansions are stable
with the stabilization parameter $s(\varphi)=p^{-l}$.
\end{theorem}

\begin{proof}
According to (\ref{19.5}), we have
$$
J_{\pi_{0},m;\varphi}(t) =\Bigl\langle
P\Big(\frac{\log_p^m|x|_p}{|x|_p}\Big)\chi_p(xt),
\varphi(x)\Bigr\rangle \qquad\qquad\qquad\qquad\qquad\qquad
$$
$$
=\int_{B_0}\frac{\log_p^m|x|_p}{|x|_p}
\big(\varphi(x)\chi_p(xt)-\varphi(0)\big)\,dx
\qquad\qquad
$$
\begin{equation}
\label{57.1}
\qquad\qquad +\int_{\bQ_p\setminus B_0}
\frac{\log_p^m|x|_p}{|x|_p}\chi_p(xt)\varphi(x)\,dx,
\end{equation}
for all $\varphi\in {\cD}(\bQ_p)$, $m=0,1,2,\dots$.

Since $\varphi\in {\cD}^l_N(\bQ_p)$, it is {\it natural\/} to
rewrite the functional $J_{\pi_{0},m;\varphi}(t)$ in the form of
the sum of integrals:
\begin{equation}
\label{57}
J_{\pi_{0},m;\varphi}(t)
=J_{\pi_{0},m;\varphi}^{1}(t)+J_{\pi_{0},m;\varphi}^{2}(t)
+\varphi(0)J_{\pi_{0},m}^{0}(t),
\end{equation}
where
\begin{equation}
\label{58.1}
J_{\pi_{0},m;\varphi}^{1}(t)
=\int_{B_l}\frac{\log_p^m|x|_p}{|x|_p}\chi_p(xt)
\big(\varphi(x)-\varphi(0)\big)\,dx,
\end{equation}
\begin{equation}
\label{58.2}
J_{\pi_{0},m;\varphi}^{2}(t) =\int_{\bQ_p\setminus B_l}
\frac{\log_p^m|x|_p}{|x|_p}\chi_p(xt)\varphi(x)\,dx,
\end{equation}
\begin{equation}
\label{58.3}
J_{\pi_{0},m}^{0}(t)
=\int_{B_l}\frac{\log_p^m|x|_p}{|x|_p}\big(\chi_p(xt)-1\big)\,dx.
\end{equation}

Since $\varphi\in {\cD}^l_N(\bQ_p)$, it is clear that
$$
\begin{array}{rcl}
\displaystyle
\frac{\log_p^m|x|_p}{|x|_p}\big(\varphi(x)-\varphi(0)\big)\Delta_{l}(x)&=&0,
\medskip \\
\displaystyle
\frac{\log_p^m|x|_p}{|x|_p}\varphi(x)\big(1-\Delta_{l}(x)\big)&\in& {\cD}^l_N(\bQ_p).
\end{array}
$$
Thus as above, according to (\ref{12.1}), for their Fourier
transforms (\ref{58.1}), (\ref{58.2}) we have
\begin{equation}
\label{59}
J_{\pi_{0},m;\varphi}^{1}(t)=J_{\pi_{0},m;\varphi}^{2}(t)=0,
\quad \forall \, |t|_p> p^{-l}.
\end{equation}

Let us calculate integral (\ref{58.3}). Suppose that
$|t|_p=p^{M}$, \ $M >-l$.

We start with the case $m=0$. Taking into account that $-M+1\le
l$, according to formulas from ~\cite[Ch.II,\S2.2.]{G-Gr-P},~\cite[IV]{Vl-V-Z},
we have
$$
J_{\pi_{0},0}^{0}(t) =\int_{B_l}\frac{\chi_p(xt)-1}{|x|_p}\,dx
=\sum_{\gamma=-\infty}^{l}
p^{-\gamma}\int_{S_{\gamma}}\big(\chi_p(xt)-1\big)\,dx
$$
$$
\quad =-p^{-(-M+1)}p^{-M+1-1}
-\sum_{\gamma=-M+1}^{l}p^{-\gamma}\Big(1-\frac{1}{p}\Big)p^{\gamma}
$$
\begin{equation}
\label{60}
=-\frac{1}{p}-\Big(1-\frac{1}{p}\Big)(l+M)
=-\frac{1}{p}-\Big(1-\frac{1}{p}\Big)\log_p\Big(\frac{|t|_p}{p^{-l}}\Big).
\end{equation}
Relations (\ref{59}) and (\ref{60}) imply that
\begin{equation}
\label{62.1}
J_{\pi_{0},0;\varphi}(t)
=\varphi(0)\bigg(-\frac{1}{p}-\Big(1-\frac{1}{p}\Big)
\log_p\Big(\frac{|t|_p}{p^{-l}}\Big)\bigg), \quad |t|_p>p^{-l}.
\end{equation}
Note that the last relation can also be proved if we use the
representation of functional (\ref{57.1}) in the form of
convolution $J_{\pi_{0},m;\varphi}(t)
=F\big[P\big(\frac{1}{|x|_p}\big)\big](t)*F\big[\varphi(x)\big](t)$,
and formula~\cite[IX,(2.8)]{Vl-V-Z}:
$$
F\big[\big(1-p^{-1}\big)\log_p|x|_p\big](t)
=-P\Big(\frac{1}{|t|_p}\Big)-p^{-1}\delta(t).
$$

In the case $m=1,2,\dots$, for $|t|_p=p^{M}$, \ $M>-l$, using
formulas from~\cite[Ch.II,\S2.2.]{G-Gr-P},~\cite[IV]{Vl-V-Z}, we obtain
$$
J_{\pi_{0},m}^{0}(t)
=\int_{B_l}\frac{\log_p^m|x|_p}{|x|_p}\big(\chi_p(xt)-1\big)\,dx
=\sum_{\gamma=-\infty}^{l}
p^{-\gamma}\gamma^m\int_{S_{\gamma}}\big(\chi_p(xt)-1\big)\,dx
$$
$$
=-p^{-(-M+1)}(-M+1)^mp^{-M+1-1}
-\sum_{\gamma=-M+1}^{l}p^{-\gamma}\gamma^m\Big(1-\frac{1}{p}\Big)p^{\gamma}
$$
\begin{equation}
\label{61}
\qquad\qquad\qquad
=-\frac{1}{p}(-M+1)^m
-\Big(1-\frac{1}{p}\Big)\sum_{\gamma=-M+1}^{l}\gamma^m.
\end{equation}
Next, using formulas (\ref{130}), (\ref{130*}), relation
(\ref{61}) can be easily transformed to the following form
$$
J_{\pi_{0},m}^{0}(t)=-\frac{1}{p}(-M+1)^m
-\Big(1-\frac{1}{p}\Big)\big({\bf S}_{m}(l)-{\bf S}_{m}(-M)\big)
\qquad
$$
$$
=-\frac{1}{p}(-M+1)^m +\Big(1-\frac{1}{p}\Big)
\frac{1}{m+1}\Big((-M)^{m+1}-l^{m+1}
$$
$$
\qquad
-C_{m+1}^1{\bf B}_{1}\big((-M)^{m}-l^{m}\big)
+\sum_{r=2}^{m}C_{m+1}^r{\bf
B}_{r}\big((-M)^{m+1-r}-l^{m+1-r}\big)\Big)
$$
$$
=-\frac{1}{p}(-1)^m(\log_p|t|_p-1)^m
\qquad\qquad\qquad\qquad\qquad\qquad
$$
$$
+\Big(1-\frac{1}{p}\Big)
\frac{1}{m+1}\Big((-1)^{m+1}\big(\log_p^{m+1}|t|_p-\log_p^{m+1}p^{-l}\big)
$$
$$
\qquad\qquad
-(-1)^mC_{m+1}^1{\bf B}_{1}\big(\log_p^{m}|t|_p-\log_p^{m}p^{-l}\big)
$$
\begin{equation}
\label{61.1}
+\sum_{r=2}^{m}(-1)^{m+1-r}C_{m+1}^r{\bf B}_{r}
\big(\log_p^{m+1-r}|t|_p-\log_p^{m+1-r}p^{-l}\big)\Big),
\end{equation}
where the Bernoulli numbers ${\bf B}_{r}$, $r=0,1,\dots,m$ are
defined by (\ref{131}), the polynomial ${\bf S}_{m}(\gamma_0)$ is
given by (\ref{130}).

Relations (\ref{57}), (\ref{59}), (\ref{61.1}) imply
$$
J_{\pi_{0},m;\varphi}(t)
=\varphi(0)\bigg\{\frac{1}{p}(-1)^{m+1}(\log_p|t|_p-1)^m
\qquad\qquad\qquad
$$
$$
+\Big(1-\frac{1}{p}\Big)
\frac{1}{m+1}\Big((-1)^{m+1}\big(\log_p^{m+1}|t|_p-\log_p^{m+1}p^{-l}\big)
$$
$$
\qquad\qquad
-(-1)^mC_{m+1}^1{\bf B}_{1}\big(\log_p^{m}|t|_p-\log_p^{m}p^{-l}\big)
$$
\begin{equation}
\label{62}
+\sum_{r=2}^{m}(-1)^{m+1-r}C_{m+1}^r{\bf B}_{r}
\big(\log_p^{m+1-r}|t|_p-\log_p^{m+1-r}p^{-l}\big)\Big)\bigg\},
\end{equation}
for all $|t|_p> p^{-l}$.

Thus relations (\ref{50}), (\ref{50.1}) hold.
\end{proof}

\begin{corollary}
\label{cor1}
If $\alpha=1$ then relations {\rm (\ref{49.0})}, {\rm (\ref{25})}
imply the statement {\rm (\ref{12.1})} of Lemma~{\rm\cite[VII.2.]{Vl-V-Z}}.
\end{corollary}

\begin{remark}
\label{rem2} \rm
Asymptotical expansion (\ref{49.1}) can be represented in the form
$$
|x|_p^{\alpha-1}\log_p^m|x|_p\chi_p(xt)
\qquad\qquad\qquad\qquad\qquad\qquad\qquad
$$
\begin{equation}
\label{49.1-11}
\qquad
=\delta(x)\frac{\log_p^{m}|t|_p}{|t|_p^{\alpha}}
\Big(\sum_{k=0}^{N}A_{k}(\alpha)\log_p^{-k}|t|_p+o\big(\log_p^{-N}|t|_p\big)\Big),
\quad |t|_p \to \infty,
\end{equation}
where $A_{k}(\alpha)$ is an explicit computable constant, $k=0,1,\dots$.
Here a {\em stabilization property} is expressed be the following assertion:
{\em for $N \ge m$ and for enough large $|t|_p$ the remainder disappears and
the asymptotic expansion turns to an exact equality}
~\footnote{We emphasize the representation (\ref{49.1-11}) after a remark
of the anonymous referee of this paper.}.

The same remark is also true for the case of asymptotical expansions (\ref{50.1}).
\end{remark}

\begin{remark}
\label{rem3} \rm
Since $\varphi\in {\cD}^l_N(\bQ_p)$, to calculate
asymptotics of the functionals $J_{\pi_{\alpha},m;\varphi}(t)$ and
$J_{\pi_{0},m;\varphi}(t)$ (for the case $\pi_1(x) \equiv 1$, $\alpha\ne0$),
it is {\it natural\/} to represent these functionals as the sums of
integrals (\ref{45.1}) and (\ref{57}), respectively. However, we
can represent these functionals as the sums of integrals
\begin{equation}
\label{46*}
{\widetilde J}_{\pi_{\alpha},m;\varphi}^{1}(t)
=\int_{B_{l_0}}|x|_p^{\alpha-1}(x)\log_p^m|x|_p\chi_p(xt)
\big(\varphi(x)-\varphi(0)\big)\,dx,
\end{equation}
\begin{equation}
\label{47*}
{\widetilde J}_{\pi_{\alpha},m;\varphi}^{2}(t)
=\int_{\bQ_p\setminus B_{l_0}}
|x|_p^{\alpha-1}\log_p^m|x|_p\chi_p(xt)\varphi(x)\,dx,
\end{equation}
\begin{equation}
\label{48*}
{\widetilde J}_{\pi_{\alpha},m}^{0}(t)
=\int_{B_{l_0}}|x|_p^{\alpha-1}\log_p^m|x|_p\chi_p(xt)\,dx,
\qquad\qquad
\end{equation}
where $l_0\in \bZ$. For example, we can choose $l_0=0$, as in the
standard representations (\ref{45}) and (\ref{57.1}). In this case
$$
\begin{array}{rcl}
\displaystyle
|x|_p^{\alpha-1}\log_p^m|x|_p\big(\varphi(x)-\varphi(0)\big)\Delta_{l_0}(x)
&\in& {\cD}^{\min(l,l_0)}_N(\bQ_p), \medskip \\
\displaystyle
|x|_p^{\alpha-1}\log_p^m|x|_p\varphi(x)\big(1-\Delta_{l_0}(x)\big)
&\in& {\cD}^{\min(l,l_0)}_N(\bQ_p).
\end{array}
$$
and, as above, according to (\ref{12.1}),
${\widetilde J}_{\pi_{\alpha},m;\varphi}^{1}(t)
={\widetilde J}_{\pi_{\alpha},m;\varphi}^{2}(t)=0$ for all
$|t|_p >p^{\max(-l,-l_0)}$. Thus repeating the above calculations
almost word for word, we obtain the asymptotic formulas from
Theorem~\ref{th2-1}. However, in this case the minimal
stabilization parameter is equal to $s(\varphi)=p^{\max(-l,-l_0)}$.

The same remark is also true for the case of Theorems~\ref{th2-2}.
\end{remark}

\section{Asymptotic formulas for singular Fourier integrals (the case $\pi_1(x) \not\equiv 1$)}
\label{s3}

Now we consider the case of distributions
$f_{\pi_{\alpha};m}(x)=|x|_p^{\alpha-1}\pi_{1}(x)\log_p^m|x|_p$, $m=0,1,2,\dots$.

\begin{theorem}
\label{th3}
Let $\varphi\in {\cD}^l_N(\bQ_p)$, and let $k_0>0$ be
the rank of the character $\pi_1(x)$. Then the functional
$J_{\pi_{\alpha},m;\varphi}(t)$ has the following asymptotical
behavior:

{\rm (a)} If $m=0$ then
\begin{equation}
\label{65}
J_{\pi_{\alpha},0;\varphi}(t) =\Bigl\langle
|x|_p^{\alpha-1}\pi_{1}(x)\chi_p(xt),\varphi(x)\Bigr\rangle
=\varphi(0)\frac{\Gamma_p(\pi_{\alpha})}{|t|_p^{\alpha}\pi_{1}(t)},
\quad |t|_p > p^{-l+k_0},
\end{equation}
for all $\varphi \in {\cD}^l_N(\bQ_p)$, where the
$\Gamma$-function $\Gamma_p(\pi_{\alpha})$ is given by {\rm (\ref{25.1})},
i.e., in the weak sense
\begin{equation}
\label{65.2}
|x|_p^{\alpha-1}\pi_{1}(x)\chi_p(xt)
=\delta(x)\frac{\Gamma_p(\pi_{\alpha})}{|t|_p^{\alpha}\pi_{1}(t)},
\quad |t|_p \to \infty.
\end{equation}

{\rm (b)} If $m=1,2,\dots$ then
$$
J_{\pi_{\alpha},m;\varphi}(t) =\Bigl\langle
|x|_p^{\alpha-1}\pi_{1}(x)\log_p^m|x|_p\chi_p(xt),
\varphi(x)\Bigr\rangle \qquad\qquad
$$
\begin{equation}
\label{65.1}
=\varphi(0)\sum_{k=0}^{m}C_{m}^{k}\log_p^{k}e\,
\frac{d^{k}\Gamma_p(\pi_{\alpha})}{d\alpha^{k}}
\frac{\log_p^{m-k}|t|_p}{|t|_p^{\alpha}\pi_{1}(t)},
\quad |t|_p >p^{-l+k_0},
\end{equation}
for all $\varphi \in {\cD}^l_N(\bQ_p)$, i.e., in the weak sense
$$
|x|_p^{\alpha-1}\pi_{1}(x)\log_p^m|x|_p\chi_p(xt)
\qquad\qquad\qquad\qquad\qquad\qquad\qquad
$$
\begin{equation}
\label{65.3}
=\delta(x)\sum_{k=0}^{m}C_{m}^{k}\log_p^{k}e\,
\frac{d^{k}\Gamma_p(\pi_{\alpha})}{d\alpha^{k}}
\frac{\log_p^{m-k}|t|_p}{|t|_p^{\alpha}\pi_{1}(t)},
\quad |t|_p \to \infty,
\end{equation}
with respect to an asymptotic sequence
$\{\pi_{\alpha+1}^{-1}(t)\log_p^{m-k}|t|_p:k=0,1,\dots,m\}$.

Thus for any $\varphi \in {\cD}(\bQ_p)$, relations {\rm
(\ref{65}), (\ref{65.1})} are exact equalities for sufficiently
big $|t|_p>p^{-l+k_0}$, i.e., these asymptotical expansions are
stable with the stabilization parameter $s(\varphi)=p^{-l+k_0}$.
\end{theorem}

\begin{proof}
Let $m=0$. Taking into account
formulas~\cite[VII,(3.10),(5.4)]{Vl-V-Z}, the functional
$J_{\pi_{\alpha},0;\varphi}(t)$ can be rewritten as a convolution:
$$
J_{\pi_{\alpha},0;\varphi}(t) =\bigl\langle
|x|_p^{\alpha-1}\pi_{1}(x)\chi_p(xt), \varphi(x)\bigr\rangle
\qquad\qquad\qquad\qquad\qquad\qquad\qquad
$$
\begin{equation}
\label{66}
=F[|x|_p^{\alpha-1}\pi_{1}(x)\varphi(x)](t)
=F[|x|_p^{\alpha-1}\pi_{1}(x)](t)*\psi(t),
\end{equation}
where $\psi(\xi)=F[\varphi(x)](\xi)$ and according
to~\cite[VIII,(2.1)]{Vl-V-Z},
\begin{equation}
\label{67}
F[|x|_p^{\alpha-1}\pi_{1}(x)](t)
=\Gamma_p(\pi_{\alpha})|t|_p^{-\alpha}\pi_{1}^{-1}(t).
\end{equation}

Since $\varphi(x)\in {\cD}^l_N(\bQ_p)$ then, in view of
(\ref{12.1}), $\psi(\xi)\in {\cD}^{-N}_{-l}(\bQ_p)$. If
$|t|_p>p^{-l}$, according to (\ref{67}), (\ref{11.1}), relation
(\ref{66}) can be rewritten as
\begin{equation}
\label{68}
J_{\pi_{\alpha},0;\varphi}(t)=\Gamma_p(\pi_{\alpha})
\int_{\bQ_p} |t-\xi|_p^{-\alpha}\pi_{1}^{-1}(t-\xi)\psi(\xi)\,d\xi.
\end{equation}
Since $|t|_p>p^{-l}$ and $\xi\in B_{-l}$ the last integral is well
defined for any $\alpha$. Moreover, we have $|t-\xi|_p=|t|_p$ for
$|t|_p>p^{-l}$, \ $\xi\in B_{-l}$. Thus relation (\ref{68}) can be
transformed to the form
\begin{equation}
\label{69}
J_{\pi_{\alpha},0;\varphi}(t)
=\Gamma_p(\pi_{\alpha})|t|_p^{-\alpha}\pi_{1}^{-1}(t)\Psi(t),
\quad |t|_p>p^{-l},
\end{equation}
where
\begin{equation}
\label{70}
\Psi(t)=\int_{B_{-l}}\pi_{1}^{-1}\Big(1-\frac{\xi}{t}\Big)\psi(\xi)\,d\xi.
\end{equation}

Let $k_0$ be the rank of the character $\pi_1(x)\not\equiv 1$. In
is clear that if $|t|_p>p^{-l+k_0}$ then the inequality
$\big|\frac{\xi}{t}\big|_p \le p^{-k_0}$ holds for all $\xi\in
B_{-l}$. Thus in view of (\ref{17}), we see that
$\pi_{1}^{-1}\big(1-\frac{\xi}{t}\big)\equiv 1$ for all $\xi\in
B_{-l}$ and $|t|_p>p^{-l+k_0}$. Next, applying an analog of the
Lebesque theorem to limiting passage under the sign of an integral
to (\ref{70}), and taking into account that $|\pi_1(x)|=1$ and
$$
\int_{B_{-l}}\psi(\xi)\,d\xi=\int_{\bQ_p}\psi(\xi)\,d\xi=\varphi(0),
$$
we see that (\ref{69}), (\ref{70}) imply relations (\ref{65}),
(\ref{65.2}) for all $|t|_p>p^{-l+k_0}$.

If $m=1,2,\dots$, differentiating (\ref{69})
with respect to $\alpha$, we obtain
$$
J_{\pi_{\alpha},m;\varphi}(t)
=\log_p^{m}\frac{d^{m}}{d\alpha^{m}}J_{\pi_{\alpha},0;\varphi}(t)
\qquad\quad\qquad\qquad\qquad\qquad\qquad\qquad\qquad
$$
\begin{equation}
\label{71}
=\sum_{k=0}^{m}C_{m}^{k}\log_p^{k}e\,\frac{d^{k}\Gamma_p(\pi_{\alpha})}{d\alpha^{k}}
|t|_p^{-\alpha}\log_p^{m-k}|t|_p\pi_{1}^{-1}(t)\Psi(t), \quad |t|_p>p^{-l}.
\end{equation}
Just as above, since $\Psi(t)=\varphi(0)$ for $|t|_p>p^{-l+k_0}$,
relation (\ref{71}) implies (\ref{65.1}), (\ref{65.3}).
\end{proof}

The analogues of Remarks~\ref{rem2},~\ref{rem3} are also true for
the case of Theorem~\ref{th3}.

\section{$p$-Adic version of the Erd\'elyi lemma}
\label{s4}

Theorems~\ref{th2-1},~\ref{th3} for $Re\alpha>0$ imply
the following $p$-adic version of the well known Erd\'elyi lemma.

\begin{corollary}
\label{cor2}
Let $k_0$ be the rank of the character
$\pi_1$, and $\varphi \in {\cD}^l_N(\bQ_p)$. Then
for $Re\alpha>0$, \ $m=0,1,2,\dots$, we have
$$
\int_{\bQ_p}|x|_p^{\alpha-1}\pi_{1}(x)\log_p^m|x|_p\chi_p(xt)
\varphi(x)\,dx
\qquad\qquad\qquad\qquad\qquad\qquad
$$
$$
\qquad\qquad
=\varphi(0)\sum_{k=0}^{m}C_{m}^{k}\log_p^{k}e\,
\frac{d^{k}\Gamma_p(\pi_{\alpha})}{d\alpha^{k}}
\frac{\log_p^{m-k}|t|_p}{|t|_p^{\alpha}\pi_{1}(t)}, \quad
|t|_p>p^{-l+k_0}.
$$
Moreover, for any $\varphi \in {\cD}(\bQ_p)$, the last relation is
a stable asymptotical expansion.
\end{corollary}

\section{Some auxiliary lemmas}
\label{s5}

\begin{lemma}
\label{lem3}
Let $\varphi(x)\in {\cD}^l_N(\bQ_p)$ and
$\psi(t)=F\big[|x|_p^{\alpha}\varphi(x)\big](t)$, \ $Re \alpha >-1$ then
\begin{equation}
\label{26}
\psi(t)=\left\{
\begin{array}{rcl}
\in {\cD}^{-N}_{-l}(\bQ_p), \quad |t|_p&\le &p^{-l}, \medskip \\
\varphi(0)\Gamma_p(\alpha+1)\frac{1}{|t|_p^{\alpha+1}},
\quad |t|_p&>&p^{-l}, \\
\end{array}
\right.
\end{equation}
where $\Gamma_p(\alpha)$ is given by formula {\rm (\ref{25})}.
\end{lemma}

\begin{proof}
Since $Re \alpha >-1$ the integral $\psi(t)$ converges absolutely.
We rewrite it as the sum $\psi(t)=\psi_1(t)+\psi_2(t)$, where
\begin{equation}
\label{27}
\psi_1(t)=\int_{B_{l}}\chi_p(t x)|x|_p^{\alpha}\varphi(x)\,dx, \quad
\psi_2(t)=\int_{\bQ_p\setminus B_{l}}\chi_p(tx)|x|_p^{\alpha}\varphi(x)\,dx.
\end{equation}

If $x\in \bQ_p\setminus B_{l}$ the function $|x|_p^{\alpha}$ has a
parameter of constancy $\ge l$, i.e., $|x|_p^{\alpha}\varphi(x)\in
{\cD}^l_N$. Hence according to (\ref{12.1}),
\begin{equation}
\label{28}
\psi_2(t)=F\big[|x|_p^{\alpha}(1-\Delta_{l}(x))\varphi(x)\big](t)
\in {\cD}^{-N}_{-l},
\end{equation}
i.e. $\psi_2(t)=0$ if $|t|_p>p^{-l}$.

Since $\varphi(x)\in {\cD}^l_N(\bQ_p)$, the function $\psi_1(t)$
can be rewritten as
$$
\psi_1(t)=\int_{B_{l}}\chi_p(tx)|x|_p^{\alpha}\,\varphi(x)\,dx
=\varphi(0)\int_{B_{l}}\chi_p(tx)|x|_p^{\alpha}\,dx.
$$
Next, according to~\cite[VII.2.,Example~9.]{Vl-V-Z} and
(\ref{25}), for $Re \alpha >-1$ we have
$$
F\big[|x|_p^{\alpha}\Delta_{l}(x)\big](t)
=\int_{B_{l}}\chi_p(tx)|x|_p^{\alpha}\,dx
\qquad\qquad\qquad\qquad\qquad\qquad\qquad
$$
\begin{equation}
\label{29}
=\frac{1-p^{-1}}{1-p^{-(\alpha+1)}}p^{l(\alpha+1)}\Delta_{-l}(t)
+\frac{\Gamma_p(\alpha+1)}{|t|_p^{\alpha+1}}\big(1-\Delta_{-l}(t)\big).
\end{equation}

To complete the proof of the lemma, it remains to use
(\ref{27})--(\ref{29}).
\end{proof}

\begin{lemma}
\label{lem4}
Let $\varphi(x)\in {\cD}^l_N(\bQ_p)$, \
$\psi(t)=F\big[|x|_p^{\alpha}\log_p^m|x|_p \varphi(x)\big](t)$,
$Re \alpha >-1$, $m=1,2,\dots$ then
\begin{equation}
\label{26.1}
\psi(t)=\left\{
\begin{array}{rcl}
\in {\cD}^{-N}_{-l}(\bQ_p), \quad |t|_p&\le &p^{-l}, \medskip \\
\varphi(0)\sum_{k=0}^{m}C_{m}^{k}\log_p^{m-k}e\,\frac{d^{m-k}\Gamma_p(\alpha+1)}{d\alpha^{m-k}}
\frac{\log_p^k|t|_p}{|t|_p^{\alpha+1}}, \quad |t|_p&>&p^{-l}. \\
\end{array}
\right.
\end{equation}
\end{lemma}

\begin{proof}
Since $Re \alpha>-1$, by differentiating the identity (\ref{29})
with respect to~$\alpha$, we derive the following identity:
$$
F\big[|x|_p^{\alpha}\log_p^m|x|_p\Delta_{l}(x)\big](t)
=\int_{B_{l}}\chi_p(tx)|x|_p^{\alpha}\log_p^m|x|_p\,dx
\qquad\qquad\qquad\qquad
$$
$$
=\big(1-p^{-1}\big)\log_p^m e \,\frac{d^m}{d\alpha^m}\bigg(
\frac{p^{l(\alpha+1)}}{1-p^{-(\alpha+1)}}\bigg)\Delta_{-l}(t)
\qquad
$$
\begin{equation}
\label{30}
+\big(1-\Delta_{-l}(t)\big)
\sum_{k=0}^{m}C_{m}^{k}\log_p^{m-k} e\,\frac{d^{m-k}\Gamma_p(\alpha+1)}{d\alpha^{m-k}}
\frac{\log_p^k|t|_p}{|t|_p^{\alpha+1}},
\end{equation}
where $C_{m}^{k}$ are binomial coefficients, $\Gamma_p(\alpha)$ is
given by formula (\ref{25}).

Next, repeating the constructions of Lemma~\ref{lem3} practically
word for word, we obtain the proof of Lemma~\ref{lem4}.
\end{proof}

Now we introduce the well-known relation, which we shall need to calculate some integrals.

Recall that the Bernoulli numbers are defined by the following recurrence relation
\begin{equation}
\label{131}
{\bf B}_{0}=1, \qquad \sum_{r=0}^{\gamma-1}C_{\gamma}^r{\bf B}_{r}=0.
\end{equation}
In particular, ${\bf B}_{1}=-\frac{1}{2}$, \ ${\bf B}_{2j-1}=0$,
$j=2,3,\dots$, \ ${\bf B}_{2}=\frac{1}{6}$, \ ${\bf B}_{4}=-\frac{1}{30}$.

\begin{proposition}
\label{pr1}
{\rm(}see~{\rm\cite{Kud})}
$$
{\bf S}_{s}(\gamma_0)=\sum_{\gamma=1}^{\gamma_0}\gamma^s
=\frac{1}{s+1}\sum_{r=0}^{s}C_{s+1}^r{\bf B}_{r}\gamma_0^{s+1-r}+\gamma_0^{s}
\qquad\qquad\qquad\qquad\qquad\qquad\qquad
$$
\begin{equation}
\label{130}
=\frac{1}{s+1}\Big(\gamma_0^{s+1}-C_{s+1}^1{\bf B}_{1}\gamma_0^{s}
+C_{s+1}^2{\bf B}_{2}\gamma_0^{s-1}+\cdots+C_{s+1}^s{\bf B}_{s}\gamma_0\Big),
\quad \gamma_0\ge 1,
\end{equation}
where ${\bf B}_{r}$ are the Bernoulli numbers, $r=0,1,\dots,s$, \ $s=0,1,2,\dots$.
\end{proposition}

One can consider the right-hand side of ${\bf S}_{s}(\gamma_0)$ as
a polynomial with respect to~$\gamma_0$.

\begin{lemma}
\label{lem5}
{\rm (~\cite{Al-Kh-Sh1},~\cite{Al-Kh-Sh5})}
If we consider ${\bf S}_{s}(\gamma_0)$ as a polynomial with respect
to~$\gamma_0$ then for $\gamma_0\le -1$, \ $s=0,1,2,\dots$, we have
\begin{equation}
\label{130*}
\sum_{\gamma=\gamma_0+1}^{0}\gamma^s=-{\bf S}_{s}(\gamma_0).
\end{equation}
\end{lemma}

\begin{proof}
To prove the lemma we rewrite the last sum by using relation
(\ref{130}) as
$$
\sum_{\gamma=\gamma_0+1}^{0}\gamma^s
=(-1)^s\sum_{\gamma=1}^{-\gamma_0-1}\gamma^s
=-\frac{1}{s+1}\Big((\gamma_0+1)^{s+1}-(-1)C_{s+1}^1{\bf B}_{1}(\gamma_0+1)^{s}
$$
\begin{equation}
\label{132} \quad +C_{s+1}^2{\bf B}_{2}(\gamma_0+1)^{s-1}+\cdots
+(-1)^sC_{s+1}^s{\bf B}_{s}(\gamma_0+1)\Big).
\end{equation}

Using (\ref{131}), it is easy to see that the coefficients of
$\gamma_0^{s+1}$, $\gamma_0^{s}$, $\gamma_0^{s-1}$ in the last sum
are equal to $1$, $C_{s+1}^1+C_{s+1}^1{\bf B}_{1}=-C_{s+1}^1{\bf B}_{1}$,
and $C_{s+1}^2+C_{s+1}^1{\bf B}_{1}C_{s}^1+C_{s+1}^2{\bf B}_{2}
=C_{s+1}^2{\bf B}_{2}$, respectively. Taking into account
relation (\ref{130}) and the relation ${\bf B}_{2k-1}=0$,
$j=2,3,\dots$, we calculate the coefficient of $\gamma_0^{s-j}$:
$$
C_{s+1}^{s-j}+C_{s+1}^1{\bf B}_{1}C_{s}^{s-j}
+C_{s+1}^2{\bf B}_{2}C_{s-1}^{s-j}-C_{s+1}^3{\bf B}_{3}C_{s-2}^{s-j}
\qquad\qquad\qquad
$$
$$
\qquad\qquad +\dots+(-1)^{j}C_{s+1}^{j}{\bf B}_{j}C_{s-j}^{s+1-j}
+(-1)^{j+1}C_{s+1}^{j+1}{\bf B}_{j+1}
$$
$$
\qquad\qquad\qquad =C_{s+1}^{j+1}\sum_{r=0}^{j}C_{j+1}^{r}{\bf B}_{r}
+C_{s+1}^{j+1}{\bf B}_{j+1}=C_{s+1}^{j+1}{\bf B}_{j+1},
$$
$j=2,3,\dots,s-1$. The coefficient of $\gamma_0^{0}$ is equal to
$\sum_{r=0}^{s}C_{s+1}^{r}{\bf B}_{r}=0$.
The lemma is thus proved.
\end{proof}

\bibliographystyle{amsplain}

\end{document}